\documentclass[conference]{IEEEtran}
\IEEEoverridecommandlockouts
\usepackage{geometry}
\usepackage{graphicx}
\usepackage{amssymb}
\usepackage{amsmath}
\usepackage{cite}
\usepackage{subfigure}
\usepackage{mathrsfs}
\usepackage[displaymath,mathlines]{lineno}
\usepackage{color}
\usepackage{tabulary}
\usepackage{multirow}
\usepackage{algpseudocode}
\usepackage{algorithm,algpseudocode}
\usepackage{algorithmicx}
\usepackage{pbox}
\usepackage{multicol}
\usepackage{lipsum}
\usepackage{amsthm}
\usepackage{relsize}
\usepackage{lipsum}
\usepackage{epstopdf}
\usepackage{mathtools}
\usepackage{calc}

\makeatletter
\newcommand{\doublewidetilde}[1]{{%
		\mathpalette\double@widetilde{#1}}}
\newcommand{\double@widetilde}[2]{%
	\sbox\z@{$\m@th#1\widetilde{#2}$}%
	\ht\z@=.5\ht\z@
	\widetilde{\box\z@}}
\makeatother

\newtheorem{theorem}{Theorem}
\newtheorem{lemma}{Lemma}
\newtheorem{corollary}{Corollary}

\newtheorem{assumption}{Assumption}

\begin{document}
	%
	\title{\huge RIS and  Cell-Free Massive MIMO: A  Marriage For Harsh Propagation Environments}

	\author{\IEEEauthorblockN{Trinh~Van~Chien$^{\ast}$,  Hien~Quoc~Ngo$^{\dagger}$, Symeon~Chatzinotas$^{\ast}$, Marco~Di~Renzo$^{\xi}$, and Bj\"{o}rn~Ottersten$^{\ast}$}
		\IEEEauthorblockA{$^{\ast}$Interdisciplinary Centre for Security, Reliability and Trust (SnT), University of Luxembourg, Luxembourg\\
			$^{\dagger}$School of Electronics, Electrical Engineering and Computer Science, Queen's University Belfast, Belfast, UK \\
			$^{\xi}$Universit\'e  Paris-Saclay,  CNRS,  CentraleSup\'elec, Laboratoire  des  Signaux  et  Syst\`emes,  France \vspace*{-0.8cm}}
		\thanks{This work of T. V. Chien, S. Chatzinotas, and B. Ottersten was supported by RISOTTI - Reconfigurable Intelligent Surfaces for Smart Cities under project FNR/C20/IS/14773976/RISOTTI. The work of H. Q. Ngo was supported by the UK Research and Innovation Future Leaders Fellowships under Grant MR/S017666/1. The work of M. Di Renzo was supported in part by the European Commission through the H2020 ARIADNE project under grant agreement number 871464 and through the H2020 RISE-6G project under grant agreement number 101017011. The long version of this paper was submitted to the IEEE Trans. Wireless Comm. \cite{Chien2021TWC}.} 
	}
	
	\maketitle
	
	
	\begin{abstract}
		This paper considers Cell-Free Massive  Multiple Input Multiple Output (MIMO) systems
		with the assistance of an RIS for enhancing the system
		performance. Distributed maximum-ratio combining (MRC) is
		considered at the access points (APs). We introduce an aggregated
		channel estimation method that provides sufficient information
		for data processing. The considered system is studied by using
		asymptotic analysis which lets the number of APs and/or the
		number of RIS elements grow large. A lower bound for the
		channel capacity is obtained for a finite number of APs and
		engineered scattering elements of the RIS, and closed-form expression for
		the uplink ergodic net throughput is formulated. In addition,
		a simple scheme for controlling the configuration of the
		RIS scattering elements is proposed. Numerical results verify the
		effectiveness of the proposed system design and the benefits of using RISs
		in Cell-Free Massive MIMO systems are quantified.
	\end{abstract}
	

	%
	\IEEEpeerreviewmaketitle

	\vspace*{-0.25cm}
	\section{Introduction}
	\vspace*{-0.15cm}
	Cell-Free Massive Multiple Input Multiple Output (MIMO) has recently been introduced to
	reduce the intercell interference of colocated Massive MIMO architectures. This is a network deployment where a large number of access points (APs) are located in a given coverage area to serve a small number of users \cite{ngo2017cell}. All the APs collaborate with each other via a backhaul network and serve all the users in the absence of cell boundaries. The system performance is enhanced in Cell-Free Massive MIMO systems because they inherit the benefits of the distributed MIMO and network MIMO architectures, but the users are also close to the APs. When each AP is equipped with a single antenna, maximum-ratio combining (MRC) results in a good net throughput for every user, while ensuring a low computational complexity and offering a distributed implementation that is convenient for scalability purposes. However, this network deployment cannot guarantee a good service under harsh propagation environments.
	
	RIS is an emerging technology that is capable of shaping
	the radio waves at the electromagnetic level without applying digital signal processing methods and requiring power amplifiers \cite{le2020robust}. Each element of the RIS scatters (e.g., reflects) the incident signal without using radio frequency chains and power amplification. Integrating an RIS into wireless networks introduces digitally controllable links that scale up with the number of engineered scattering elements of the RIS, whose estimation
	is, however, challenged by the lack of digital signal processing units at the RIS. For simplicity, the main attention has so far been concentrated on designing the phase shifts with perfect channel state information (CSI) \cite{zhao2020intelligent,perovic2021achievable} and the references therein. As far as the integration of Cell-Free Massive MIMO
	and RIS is concerned, recent works have formulated and
	solved optimization problems with different communication
	objectives under the assumption of perfect (and instantaneous) CSI \cite{zhou2020achievable,zhang2020capacity}. Recent results have shown that designs for the phase shifts of the RIS elements based on statistical CSI may be of practical interest and provide good performance \cite{van2021outage}. In the depicted context, no prior work has analyzed an RIS-assisted Cell-Free Massive MIMO system in the presence of spatially-correlated channels.
	
	In this work, we consider an RIS-assisted Cell-Free Massive MIMO under spatially correlated channels. We exploit a channel estimation scheme that estimates the aggregated channels including both the direct and indirect links. We analytically
	show that, even by using a low complexity MRC technique,
	the non-coherent interference, small-scale fading effects, and additive noise are averaged out when the number of APs and RIS elements increases. The received signal includes, hence, only the desired signal and the coherent interference. We derive a closed-form expression of the net throughput for the uplink data transmission. The impact of the array gain, coherent joint transmission, channel estimation errors, pilot contamination, spatial correlation, and phase shifts of the RIS, which determine the system performance, are explicitly observable in the obtained analytical expressions. With the aid of numerical simulations, we verify the effectiveness of the proposed channel estimation scheme and the accuracy of the closed-form expression of the net throughput. The obtained numerical results show that the use of RISs enhance the net throughput per user significantly, especially when the direct links are blocked with high probability.   
	
	\textit{Notation}: Upper and lower bold letters denote matrices and vectors. The identity matrix of size $N \times N$ is denoted by $\mathbf{I}_N$. $(\cdot)^{\ast},$ $(\cdot)^T,$ and $(\cdot)^H$ are the complex conjugate, transpose, and Hermitian transpose. $\mathbb{E}\{ \cdot\}$ and $\mathsf{Var} \{ \cdot \}$ denote the expectation and variance of a random variable. The circularly symmetric Gaussian distribution is denoted by  $\mathcal{CN}(\cdot, \cdot)$ and  $\mathrm{diag} (\mathbf{x})$ is the diagonal matrix whose main diagonal is given by $\mathbf{x}$. $\mathrm{tr}(\cdot)$ is the trace operator. The Euclidean norm of vector $\mathbf{x}$ is $\| \mathbf{x}\|$, and $\| \mathbf{X} \|$ is the spectral norm of matrix $\mathbf{X}$. Finally, $\mathrm{mod}(\cdot,\cdot)$ is the modulus operation and $\lfloor \cdot \rfloor$ denotes the truncated argument.
	\begin{figure}[t]
		\centering
		\includegraphics[trim=3.0cm 2.4cm 5.7cm 6.8cm, clip=true, width=2.6in]{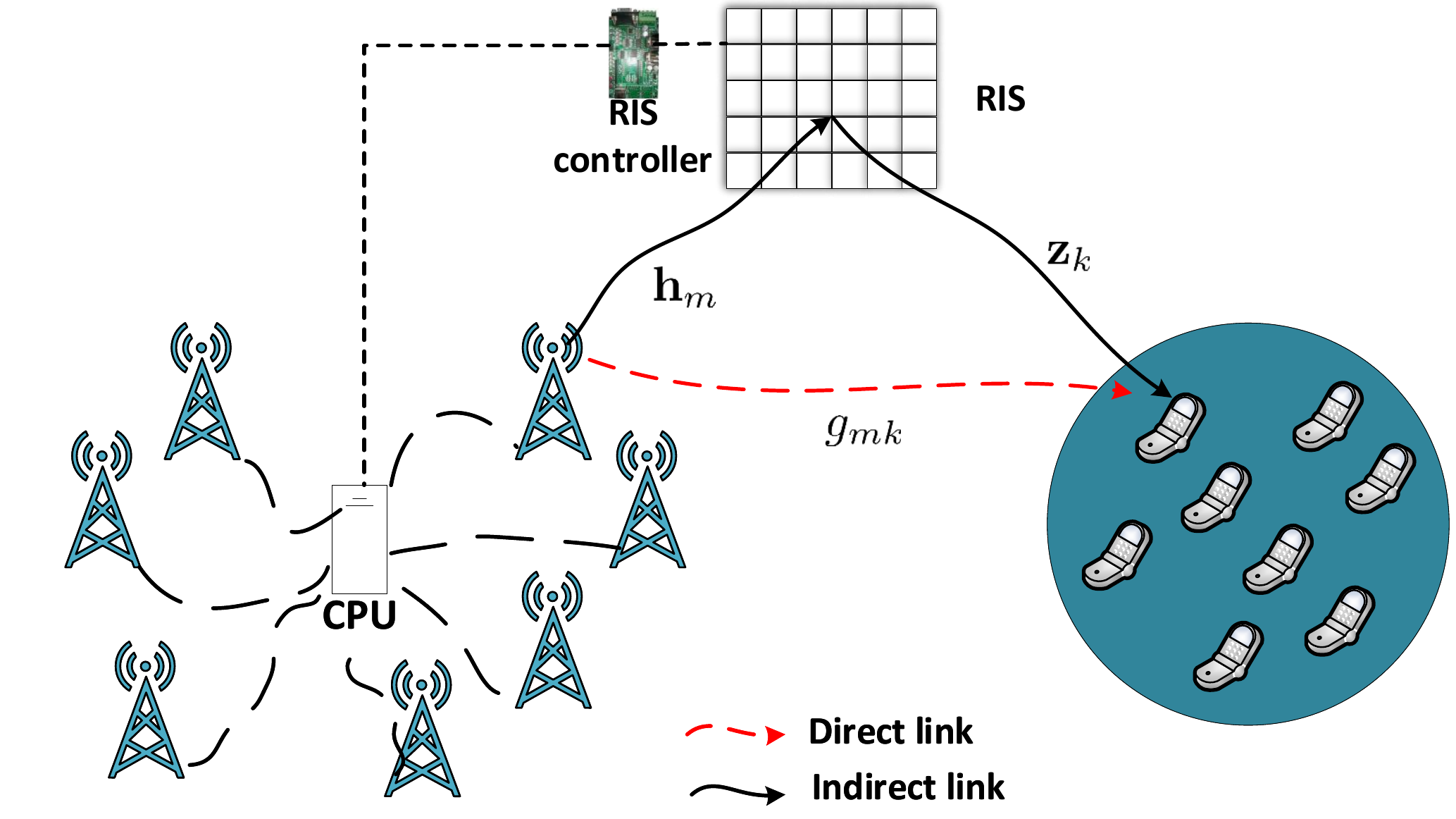} \vspace*{-0.2cm}
		\caption{An RIS-assisted Cell-Free Massive MIMO system where $M$ APs collaborate with each other to serve $K$ distant users.}
		\label{FigSysModel}
		\vspace*{-0.5cm}
	\end{figure}
	\vspace*{-0.5cm}
	\section{System Model, Channel Estimation, and RIS Phase-Shift Control} \label{Sec:SysModel}
	\vspace*{-0.2cm}
	We consider an RIS-assisted Cell-Free Massive MIMO
	system, where $M$ APs connected to a central processing
	unit (CPU) serve $K$ users on the same time and frequency
	resource. All APs and users are equipped with a single
	antenna and they are randomly located in the coverage area. The communication is assisted by an RIS that comprises $N$ engineered 
	scattering elements that can modify the phases of the incident signals. The matrix of phase shifts of the RIS is denoted by $\pmb{\Phi} = \mathrm{diag} \left( [ e^{j\theta_1}, \ldots, e^{j\theta_N}]^T \right)$, where $\theta_n \in [-\pi, \pi]$ is the phase
	shift applied by the $n$-th element of the RIS.
	\vspace*{-0.2cm}
	\subsection{Channel Model}
	\vspace*{-0.2cm}
	We assume a quasi-static block fading model where each coherence interval comprises $\tau_c$ symbols. The APs have knowledge of only the channel statistics instead of the instantaneous channel realizations. Also, $\tau_p$ symbols ($\tau_p < \tau_c$) in each coherence interval are dedicated to the channel estimation and the remaining $(\tau_c - \tau_p)$ symbols are the data transmission.
	
	The following notation is used: $g_{mk}$ is the channel between
	the user~$k$ and the AP~$m$, which is the direct link \cite{wu2019intelligent};
	$\mathbf{h}_m \in \mathbb{C}^{N}$ is the channel between the AP~$m$ and the RIS; and $\mathbf{z}_{k} \in \mathbb{C}^{N}$ is the channel between the RIS and the user~$k$. In this paper, we consider a realistic channel model by taking into account the spatial correlation among the scattering elements of the RIS, which is due to their sub-wavelength size, sub-wavelength inter-distance, and geometric layout. In an isotropic propagation environment, in particular, $g_{mk}, \mathbf{h}_m$, and $\mathbf{z}_k$ can be modeled as follows
	\begin{equation} \label{eq:Channels}
		g_{mk} \sim \mathcal{CN}(0, \beta_{mk}), \mathbf{h}_{m} \sim \mathcal{CN}(\mathbf{0}, \mathbf{R}_{m}), \mathbf{z}_{k}  \sim \mathcal{CN}(\mathbf{0}, \widetilde{\mathbf{R}}_{k}),
	\end{equation}
	where $\beta_{mk}$ is the large-scale fading coefficient; $\mathbf{R}_{m} \in \mathbb{C}^{N \times N}$
	and $\widetilde{\mathbf{R}}_{k} \in \mathbb{C}^{N \times N}$ are the covariance matrices. The covariance
	matrices in \eqref{eq:Channels} correspond to a general model, which can be further particularized for application to typical RIS designs and propagation environments. A correlation model that is applicable to isotropic scattering with uniformly distributed multipath components in the half-space in front of the RIS was recently reported in \cite{bjornson2020rayleigh}, whose covariance matrices are
	\begin{equation}\label{eq:CovarMa}
		\mathbf{R}_m = \alpha_{m} d_Hd_V\mathbf{R} \mbox{ and } \widetilde{\mathbf{R}}_{k} = \tilde{\alpha}_{k} d_Hd_V\mathbf{R}, 
	\end{equation}
	where $\alpha_{m}, \tilde{\alpha}_{mk} \in \mathbb{C}$ are the large-scale channel coefficients. The matrices in \eqref{eq:CovarMa} assume that the size of each RIS element
	is $d_H \times d_V$, with $d_H$ being the horizontal width and $d_V$ being the vertical height of each RIS element. In particular, the $(l,t)$-th element of the spatial correlation matrix $\mathbf{R} \in \mathbb{C}^{N \times N }$ in \eqref{eq:CovarMa} is $[\mathbf{R}]_{lt}= \mathrm{sinc} (2 \|\mathbf{u}_{l} - \mathbf{u}_{t} \|/ \lambda)$, where $\lambda$ is the wavelength and $\mathrm{sinc}(x) = \sin(\pi x) / (\pi x)$ is the $\mathrm{sinc}$ function. The vector $\mathbf{u}_{x}, x \in \{l,t\}$ is given by $\mathbf{u}_{x} = [0, \mod(x-1,N_H)d_H, \lfloor (x-1)/N_H\rfloor d_V]^T$, where $N_H$ and $N_V$
	denote the total number of RIS elements in each row and column, respectively.
	\vspace*{-0.15cm}
	\subsection{Uplink Pilot Training Phase}
	\vspace*{-0.15cm}
	The channels are independently estimated from the $\tau_p$ pilot
	sequences transmitted by the $K$ users. All the users share the
	same $\tau_p$ pilot sequences. In particular, $\pmb{\phi}_k \in \mathbb{C}^{\tau_p}$ with $\| \pmb{\phi}_k \|^2 = 1$ is defined as the pilot sequence allocated to the user~$k$. We
	denote by $\mathcal{P}_k$ the set of indices of the users (including the user~$k$) that share the same pilot sequence as the user~$k$. The pilot
	sequences are assumed to be mutually orthogonal such that
	the pilot reuse pattern is $\pmb{\phi}_{k'}^H \pmb{\phi}_k =  1$,  $k' \in \mathcal{P}_k$. Otherwise, $\pmb{\phi}_{k'}^H \pmb{\phi}_k =  0$. During the pilot training phase, all the $K$ users
	transmit the pilot sequences to the $M$ APs simultaneously. In
	particular, the user~$k$ transmits the pilot sequence $\sqrt{\tau_p} \pmb{\phi}_k$. The
	received training signal at the AP~$m$ can be written as 
	\begin{equation} \label{eq:ReceivedPilot}
		\mathbf{y}_{pm} = \sum_{k=1}^K \sqrt{p \tau_p}  g_{mk} \pmb{\phi}_k + \sum_{k=1}^K \sqrt{p \tau_p} \mathbf{h}_{m}^H \pmb{\Phi} \mathbf{z}_{k} \pmb{\phi}_k  + \mathbf{w}_{pm},
	\end{equation}
	where $p$ is the normalized signal-to-noise ratio (SNR) of each
	pilot symbol, and $\mathbf{w}_{pm} \in \mathbb{C}^{\tau_p}$ is the additive noise at the AP~$m$, which is distributed as $\mathbf{w}_{pm} \sim \mathcal{CN} (\mathbf{0}, \mathbf{I}_{\tau_p})$. In order
	for the AP~$m$ to estimate the desired channels from the user~$k$,
	the received training signal in \eqref{eq:ReceivedPilot} is projected on $\pmb{\phi}_k^H$
	as
	\begin{equation} \label{eq:ReceivedPilotv1}
		\begin{split}
			y_{pmk} &= \pmb{\phi}_k^H \mathbf{y}_{pm} =  \sqrt{p \tau_p}  \left(g_{mk} +  \mathbf{h}_{m}^H \pmb{\Phi} \mathbf{z}_{k} \right) \\
			&+ \sum_{k' \in \mathcal{P}_k \setminus \{k\} } \sqrt{p\tau_p}  \left(g_{mk'} +  \mathbf{h}_{m}^H \pmb{\Phi} \mathbf{z}_{k'} \right) + w_{pmk},
		\end{split}
	\end{equation}
	where $w_{pmk} = \pmb{\phi}_k^H \mathbf{w}_{pm} \sim \mathcal{CN}(0, 1)$. We emphasize that
	the co-existence of the direct and indirect channels due to the
	presence of the RIS results in a complicated channel estimation
	process. In particular, the cascaded channel in \eqref{eq:ReceivedPilotv1} results in a nontrivial procedure to apply the minimum mean-square error
	(MMSE) estimation method, as reported in previous works, for
	processing the projected signals \cite{ngo2017cell,9136914}. Based on the specific
	signal structure in \eqref{eq:ReceivedPilotv1}, we denote the channel between the
	AP~$m$ and the user~$k$ through the RIS as
	\begin{equation} \label{eq:umk}
		u_{mk} = g_{mk} + \mathbf{h}_{m}^H \pmb{\Phi} \mathbf{z}_{k}, 
	\end{equation}
	which is referred to as the \textit{aggregated channel} that comprises
	the direct and indirect link between the user~$k$ and the AP~$m$.
	By capitalizing on the definition of the aggregated channel
	in \eqref{eq:umk}, the required channels can be estimated in an effective
	manner even in the presence of the RIS. In particular, the
	aggregated channel in \eqref{eq:umk} is given by the product of weighted
	complex Gaussian and spatially correlated random variables,
	as given in \eqref{eq:Channels}. Conditioned on the phase shifts, we employ the
	linear MMSE method for estimating $u_{mk}$ at the AP. Despite
	the complex structure of the RIS-assisted channels, Lemma~\ref{lemma:ChannelEst}
	provides analytical expressions of the estimated channels.
	\begin{lemma} \label{lemma:ChannelEst}
		By assuming that the AP~$m$ employs the linear MMSE estimation method based on the observation in \eqref{eq:ReceivedPilotv1}, the estimate of the aggregate channel ${u}_{mk}$ is formulated as
		\begin{equation} \label{eq:ChannelEst}
			\hat{u}_{mk} =  \big(\mathbb{E}\{ y_{pmk}^\ast u_{mk} \} y_{pmk} \big)/ \mathbb{E} \{ | y_{pmk} |^2 \}   = c_{mk} y_{pmk},
		\end{equation}
		where $c_{mk} =  \mathbb{E}\{ y_{pmk}^\ast u_{mk} \} / \mathbb{E} \{ | y_{pmk} |^2 \}$ has the following closed-form expression
		\begin{equation} \label{eq:cmk}
			c_{mk} =  \frac{\sqrt{p\tau_p} \big( \beta_{mk} +  \mathrm{tr} \big( \pmb{\Phi}^H \mathbf{R}_{m} \pmb{\Phi} \widetilde{\mathbf{R}}_{k} \big)\big) }{p\tau_p \sum_{k' \in \mathcal{P}_k} \big( \beta_{mk'} +  \mathrm{tr} \big( \pmb{\Phi}^H \mathbf{R}_{m} \pmb{\Phi} \widetilde{\mathbf{R}}_{k'} \big)\big) + 1}.
		\end{equation}
		The estimated channel in \eqref{eq:ChannelEst} has zero mean and variance $\gamma_{mk}$ equal to
		\begin{equation} \label{eq:gammamk}
			\gamma_{mk} = \mathbb{E} \{ |\hat{u}_{mk}|^2 \} = \sqrt{p\tau_p}  \big(\beta_{mk} + \mathrm{tr} \big( \pmb{\Phi}^H \mathbf{R}_{m} \pmb{\Phi} \widetilde{\mathbf{R}}_{k} \big) \big)c_{mk}.
		\end{equation}
		Also, the channel estimation error $e_{mk} = u_{mk} - \hat{u}_{mk}$ and the channel estimate $\hat{u}_{mk}$ are uncorrelated. The channel estimation error has zero mean and variance equal to
		\begin{equation} \label{eq:EstError}
			\mathbb{E}\big\{ |e_{mk} |^2 \big\} = \beta_{mk} +  \mathrm{tr} \big( \pmb{\Phi}^H \mathbf{R}_{m} \pmb{\Phi} \widetilde{\mathbf{R}}_{k} \big) - \gamma_{mk}.
		\end{equation}
	\end{lemma}
	\begin{proof}
		It is similar to the proof in \cite{Kay1993a}, and is obtained by
		applying similar analytical steps to the received signal in \eqref{eq:ReceivedPilotv1}
		and by taking into account the structure of the RIS-assisted
		channel and the spatial correlation matrices in \eqref{eq:Channels}.
	\end{proof}
	Lemma~\ref{lemma:ChannelEst} shows that, by assuming $\pmb{\Phi}$ fixed, the aggregated channel in \eqref{eq:umk} can be estimated without increasing the pilot
	training overhead, as compared to a conventional Cell-Free
	Massive MIMO system. The obtained channel estimate in
	\eqref{eq:ChannelEst} unveils the relation $\hat{u}_{mk'}  = \frac{c_{mk'}}{c_{mk}}\hat{u}_{mk}$ if the user~$k'$
	uses the same pilot sequence as the user~$k$. Because of pilot
	contamination, it may be difficult to distinguish the signals of
	these two users. In the following, the analytical expression of
	the channel estimates in Lemma~\ref{lemma:ChannelEst} are employed for signal
	detection in the uplink data transmission. They are used also
	to optimize the phase shifts of the RIS in order to minimize
	the channel estimation error and to evaluate the corresponding
	ergodic net throughput.
	\vspace{-0.2cm}
	\subsection{RIS Phase-Shift Control}
	\vspace{-0.1cm}
	Channel estimation is a critical aspect in Cell-Free Massive
	MIMO. As discussed in previous text, in many scenarios,
	non-orthogonal pilots have to be used. This causes pilot
	contamination, which may reduce the system performance
	significantly. In this section, we design an RIS-assisted phase
	shift control scheme that is aimed to improve the quality of
	channel estimation. To this end, we introduce the normalized
	mean square error (NMSE) of the channel estimate of the
	user~$k$ at the AP~$m$ as follows
	\begin{equation} 
		\begin{split}
			& \mathrm{NMSE}_{mk} = \mathbb{E}\{|e_{mk}|^2\}/\mathbb{E}\{|u_{mk}|^2\} \\
			&=1 - \frac{p \tau_p \big(\beta_{mk} + \mathrm{tr}(\pmb{\Phi}^H \mathbf{R}_m  \pmb{\Phi} \widetilde{\mathbf{R}}_{k} ) \big)}{p \tau_p \sum_{k' \in \mathcal{P}_{k}}\big( \beta_{mk'} + \mathrm{tr}(\pmb{\Phi}^H \mathbf{R}_m  \pmb{\Phi} \widetilde{\mathbf{R}}_{k'} ) \big) + 1}.
		\end{split}
	\end{equation}
	\vspace{-0.1cm}
	where the last equality is obtained from \eqref{eq:EstError}. We optimize the
	phase shift matrix $\pmb{\Phi}$ of the RIS so as to minimize the total
	NMSE obtained from all the users and all the APs as follows
	\begin{equation} \label{Prob:NMSEk}
		\begin{aligned}
			& \underset{\{ \theta_n \} }{\mathrm{minimize}}
			&&   \sum_{m=1}^M \sum_{k=1}^K \mathrm{NMSE}_{mk} \\
			& \,\,\mathrm{subject \,to}
			& & -\pi \leq \theta_n \leq \pi, \forall n.
		\end{aligned}
	\end{equation}
	The optimal phase shifts solution to problem~\eqref{Prob:NMSEk} is obtained
	by exploiting the statistical CSI that include the large-scale
	fading coefficients and the covariance matrices. Problem~\eqref{Prob:NMSEk}
	is a fractional program, whose globally-optimal solution is
	not simple to be obtained for an RIS with a large number
	of independently tunable elements. Nonetheless, in the special network setup where the direct links from the APs to the users
	are weak enough to be negligible with respect to the RIS-assisted
	links, the optimal solution to problem~\eqref{Prob:NMSEk} is available
	in a closed-form expression as summarized in Corollary~\ref{corollary:EqualPhase}.
	\begin{corollary} \label{corollary:EqualPhase}
		If the direct links are weak enough to be negligible and the RIS-assisted channels are spatially correlated as formulated in \eqref{eq:CovarMa}, the optimal maximizer of the optimization problem in \eqref{Prob:NMSEk} is $\theta_1 = \ldots = \theta_N$, i.e., the equal phase shift design is optimal.
	\end{corollary}
	\begin{proof}
		The proof follows by analyzing the objective function
		of problem~\eqref{Prob:NMSEk} with respect to the phase-shift elements. The detailed proof is available in the journal version \cite{Chien2021TWC}.
	\end{proof}
	Corollary~\ref{corollary:EqualPhase} provides a simple but effective option to design the phase shifts of the RIS while ensuring the optimal
	estimation of the aggregated channels according to the sum-
	NMSE minimization criterion, provided that the direct link are
	completely blocked and the spatial correlation model in \eqref{eq:CovarMa}
	holds true. Therefore, an efficient channel estimation protocol
	can be designed even in the presence of an RIS with a
	large number of engineered scattering elements. The numerical results
	in Section~\ref{Sec:NumRes} show that the phase shift design in Corollary~\ref{corollary:EqualPhase}
	offers good gains in terms of net throughput even if the direct
	links are not negligible. 
	
	\vspace*{-0.2cm}
	\section{Uplink Data Transmission and Performance Analysis With MR Combining}\label{Sec:UL}
	\vspace*{-0.1cm}
	In this section, we introduce a procedure to detect the
	uplink transmitted signals and derive an asymptotic closed-form
	expression of the ergodic net throughput.
	\vspace*{-0.2cm}
	\subsection{Uplink Data Transmission Phase}
	\vspace*{-0.15cm}
	In the uplink, all the $K$ users transmit their data to the
	$M$ APs simultaneously. Specifically, the user $k$ transmits a
	modulated symbol $s_k$ with $\mathbb{E}\{|s_k|^2\} =1$. This symbol is
	weighted by a power control factor $\sqrt{\eta_k}$, $0 \leq \eta_k \leq 1$. Then, the received baseband signal, $y_{um} \in \mathbb{C},$ at the AP~$m$ is
	\begin{equation} \label{eq:yum}
		y_{m}  = \sqrt{\rho} \sum_{k=1}^K \sqrt{\eta_{k}} u_{mk} s_k + w_{m},
	\end{equation}
	where $\rho$ is the normalized uplink SNR of each data symbol
	and $w_{m}$ is the normalized additive noise with $w_{m} \sim \mathcal{CN}(0,1)$.
	For data detection, the MRC method is used at the CPU, i.e.,
	$\hat{u}_{mk}, \forall m,k,$ in \eqref{eq:ChannelEst} is employed to detect the data transmitted
	by the user~$k$. In mathematical terms, the corresponding
	decision statistic is
	\begin{equation} \label{eq:ruk}
		r_{k} = \sqrt{\rho} \sum_{m=1}^M \sum_{k'=1}^K  \sqrt{\eta_k} \hat{u}_{mk}^\ast u_{mk'} s_{k'} + \sum_{m=1}^M \hat{u}_{mk}^\ast w_{m}.
	\end{equation}
	Based on the observation $r_{k}$, the uplink ergodic net throughput
	of the user~$k$ is analyzed in the next subsection. 
	\vspace*{-0.2cm}
	\subsection{Asymptotic Analysis} \label{subsec:Asymul}
	\vspace*{-0.15cm}
	Since the number of APs, $M$, and the number of tunable
	elements of the RIS, $N$, can be large, we analyze the performance
	of two case studies: $(i)$ $N$ is fixed and $M$ is large;
	and $(ii)$ both $N$ and $M$ are large. The
	asymptotic analysis is conditioned upon a given setup of the
	CSI. To this end, the uplink weighted signal in \eqref{eq:ruk} is split
	into three terms based on the pilot reuse set $\mathcal{P}_k$, as follows
	\begin{equation} \label{eq:rukv1}
		\begin{split}
			& r_{k} =  \underbrace{\sqrt{\rho} \sum_{k' \in \mathcal{P}_k } \sum_{m=1}^M  \sqrt{\eta_{k'}} \hat{u}_{mk}^\ast u_{mk'} s_{k'}}_{\mathcal{T}_{k1}} + \\
			& \underbrace{\sqrt{\rho} \sum_{k' \notin \mathcal{P}_k} \sum_{m=1}^M   \sqrt{\eta_{k'}} \hat{u}_{mk}^\ast u_{mk'} s_{k'}}_{\mathcal{T}_{k2}} + \underbrace{\sum_{m=1}^M \hat{u}_{mk}^\ast w_{m}}_{\mathcal{T}_{k3}},
		\end{split}
	\end{equation}
	where $\mathcal{T}_{k1}$ accounts for the signals received from all the users
	in $\mathcal{P}_k$, and $\mathcal{T}_{k2}$ accounts for the mutual interference from the
	users that are assigned orthogonal pilot sequences. The impact
	of the additive noise obtained after applying MR combining is
	given by $\mathcal{T}_{k3}$. From \eqref{eq:ReceivedPilotv1}-\eqref{eq:ChannelEst}, we obtain the following identity
	\begin{equation} \label{eq:Termv1}
		\begin{split}
			& \sum_{m=1}^M \sqrt{\eta_{k'}} \hat{u}_{mk}^\ast u_{mk'}  =  \sum_{k'' \in \mathcal{P}_k \setminus \{k' \} } \sum_{m=1}^M \sqrt{\eta_{k'}p\tau_p}  c_{mk} u_{mk'}  u_{mk''}^{\ast}   \\
			&+ \sum_{m=1}^M \sqrt{\eta_{k'} p \tau_p}  c_{mk} |u_{mk'}|^2 + \sum_{m=1}^M \sqrt{\eta_{k'}} c_{mk} u_{mk'}   w_{pmk}^\ast,
		\end{split}
	\end{equation}
	\subsubsection{Case I}
	$N$ is fixed and $M$ is large, i.e.,  $M \rightarrow \infty$. In
	this case, we divide both sides of \eqref{eq:Termv1} by $M$ and exploits
	Tchebyshev’s theorem \cite{cramer2004random}\footnote{Let $X_1, \ldots, X_n$ be independent random variables such that $\mathbb{E}\{ X_i \} = \bar{x}_i$ and $\mathsf{Var}\{ X_i\} \leq c < \infty$. Then, Tchebyshev's theorem states $\frac{1}{n}\sum_{n'=1}^n X_{n'} \xrightarrow[n \rightarrow \infty]{P} \frac{1}{n} \sum_{n'} \bar{x}_{n'}.$} to obtain
	\begin{multline} \label{eq:Tchev1}
		\frac{1}{M} \sum_{m=1}^M \sqrt{\eta_{k'}} \hat{u}_{mk}^\ast u_{mk'} \xrightarrow[M \rightarrow \infty ]{P} \\ 
		\frac{1}{M} \sum_{m=1}^M \sqrt{\eta_{k'} p \tau_p}  c_{mk} \big(\beta_{mk'} + \mathrm{tr}\big( \pmb{\Phi}^H \mathbf{R}_m  \pmb{\Phi} \widetilde{\mathbf{R}}_{k'} \big) \big),
	\end{multline}
	where $\xrightarrow{P}$ denotes the convergence in probability.\footnote{A sequence $\{ X_n \}$ of random variables converges in probability to the random variable $X$ if, for all $\epsilon > 0$, it holds that $\lim_{n \rightarrow \infty} \mathrm{Pr}(|X_n - X| > \epsilon ) = 0$, where $\mathrm{Pr}(\cdot)$ denotes the probability of an event.} Note that the second and third terms in \eqref{eq:Termv1} converge to zero. By inserting \eqref{eq:Tchev1} into the decision variable in \eqref{eq:rukv1}, we obtain the following deterministic value
	\begin{multline} \label{eq:Asympt1}
		\frac{1}{M}r_{k} \xrightarrow[M\rightarrow \infty]{P} \\
		\frac{1}{M} \sum_{k' \in \mathcal{P}_k} 
		\sum_{m=1}^M \sqrt{\eta_{k'} p \tau_p \rho_u}  c_{mk} \big(\beta_{mk'} + \mathrm{tr}\big( \pmb{\Phi}^H \mathbf{R}_m  \pmb{\Phi} \widetilde{\mathbf{R}}_{k'} \big) \big) s_{k'},
	\end{multline}
	because $\mathcal{T}_{k2}/M \rightarrow 0$ and $\mathcal{T}_{k3}/M \rightarrow 0 $ as $M \rightarrow \infty$.
	The result in \eqref{eq:Asympt1} unveils that, for a fixed $N$, the channels
	become asymptotically orthogonal. In particular, the small scale
	fading, the non-coherent interference, and the additive
	noise vanish. The only residual impairment is the pilot contamination
	caused by the users that employ the same pilot
	sequence. Due to pilot contamination, the system performance
	cannot be improved by adding more APs if MRC is used.
	The contributions of both the direct and RIS-assisted indirect
	channels appear explicitly in \eqref{eq:Asympt1} through $\beta_{mk}$ and $\mathrm{tr}(\pmb{\Phi}^H \mathbf{R}_m \pmb{\Phi} \widetilde{\mathbf{R}}_{k'})$, respectively.
	\subsubsection{Case II}
	Both $N$ and $M$ are large, i.e., $N\rightarrow \infty$ and $M \rightarrow \infty$. We first need some assumptions on the covariance
	matrices $\mathbf{R}_m$ and $\widetilde{\mathbf{R}}_{k}$, as summarized as follows.
	\begin{assumption} \label{Assumption1}
		For $m= 1,\ldots,M$ and $k=1,\ldots,K,$  the covariance matrices $\mathbf{R}_m$ and $\widetilde{\mathbf{R}}_{k}$ are assumed to fulfill the following properties
		\begin{align}
			&\underset{N}{\limsup} \, \| \mathbf{R}_m\|_2 < \infty, \underset{N}{\liminf} \, \frac{1}{N} \mathrm{tr} ( \mathbf{R}_m) > 0, \label{eq:Asymp1}\\
			&\underset{N}{\limsup} \, \| \widetilde{\mathbf{R}}_{k} \|_2 < \infty, \underset{N}{\liminf} \,  \frac{1}{N} \mathrm{tr} ( \widetilde{\mathbf{R}}_{k}) > 0. \label{eq:Asymp2v1}
		\end{align}
	\end{assumption}
	The assumptions in \eqref{eq:Asymp1} and \eqref{eq:Asymp2v1} imply that the largest
	singular value and the sum of the eigenvalues (counted with
	their mutiplicity) of the $N \times N$ covariance matrices that
	characterize the spatial correlation among the channels of the
	RIS elements are finite and positive. Dividing both sides of
	\eqref{eq:Termv1} by $MN$ and applying Tchebyshev’s theorem, we obtain
	\begin{multline}\label{eq:AsymMNUL}
		\frac{1}{MN} \sum_{m=1}^M \sqrt{\eta_{k'}} \hat{u}_{mk}^\ast u_{mk'} \xrightarrow[\substack{M \rightarrow \infty\\ N \rightarrow \infty} ]{P} \\
		\frac{1}{MN} \sum_{m=1}^M \sqrt{\eta_{k'} p\tau_p \rho_u}  c_{mk}  \mathrm{tr}\big( \pmb{\Phi}^H \mathbf{R}_m  \pmb{\Phi} \widetilde{\mathbf{R}}_{k'} \big).
	\end{multline}
	We observe that $\pmb{\Phi}^H \mathbf{R}_m \pmb{\Phi} \widetilde{\mathbf{R}}_{k'}$ is similar to $\widetilde{\mathbf{R}}_{k'}^{1/2} \pmb{\Phi} \mathbf{R}_{m} \pmb{\Phi}^H \widetilde{\mathbf{R}}_{k'}^{1/2}$, which is a positive semi-definite
	matrix.\footnote{Two matrices $\mathbf{A}$ and $\mathbf{B}$ of size $N \times N$ are similar if there exists an invertible $N \times N$ matrix $\mathbf{U}$ such that $\mathbf{B} = \mathbf{U}^{-1} \mathbf{A} \mathbf{U}$.} Because similar matrices have the same eigenvalues, it
	follows that $\mathrm{tr}\big( \pmb{\Phi}^H \mathbf{R}_m \pmb{\Phi} \widetilde{\mathbf{R}}_{k'} \big) >0$. Based on Assumption~\ref{Assumption1},
	we obtain the following inequalities
	\begin{equation} \label{eq:traceConv}
		\begin{split}
			& \frac{1}{N}\mathrm{tr}\big(\pmb{\Phi}^H \mathbf{R}_m \pmb{\Phi} \widetilde{\mathbf{R}}_{k'} \big) \stackrel{(a)}{\leq} \frac{1}{N} \| \pmb{\Phi} \|_2 \mathrm{tr}\big( \mathbf{R}_m \pmb{\Phi} \widetilde{\mathbf{R}}_{k'} \big) \\
			& \stackrel{(b)}{=}  \frac{1}{N}\mathrm{tr}\big(  \pmb{\Phi} \widetilde{\mathbf{R}}_{k'} \mathbf{R}_m \big) \stackrel{(c)}{\leq} \frac{1}{N} \| \widetilde{\mathbf{R}}_{k'} \|_2 \mathrm{tr}(\mathbf{R}_{m}),
		\end{split}
	\end{equation}
	where $(a)$ is obtained by an inequality on the trace of the product of
	matrices; $(b)$ follows because $\|\pmb{\Phi}\|_2 = 1$; and $(c)$ is obtained from  Assumption~\ref{Assumption1}. Based on
	Assumption~\ref{Assumption1}, in addition, the last inequality in \eqref{eq:rukv1} is bounded by a
	positive constant. From \eqref{eq:AsymMNUL} and \eqref{eq:traceConv}, therefore, the decision
	variable in \eqref{eq:rukv1} can be formulated as
	\begin{multline} \label{eq:Asymp2}
		\frac{1}{MN}r_{k} \xrightarrow[\substack{M\rightarrow \infty\\ N \rightarrow \infty}]{P} \\
		\frac{1}{MN} \sum_{k' \in \mathcal{P}_k} \sum_{m=1}^M \sqrt{\eta_{k'} p\tau_p \rho_u}  c_{mk} \mathrm{tr}\big( \pmb{\Phi}^H \mathbf{R}_m  \pmb{\Phi} \widetilde{\mathbf{R}}_{k'} \big) s_{k'}.
	\end{multline}
	The expression obtained in \eqref{eq:Asymp2} reveals that, as $M,N \rightarrow \infty$,
	the post-processed signal at the CPU consists of the desired
	signal of the intended user $k$ and the interference from the
	other users in $\mathcal{P}_k$. Compared with \eqref{eq:Asympt1}, we observe that \eqref{eq:Asymp2} is independent of the direct links and depends only on
	the RIS-assisted indirect links. This highlights the potentially
	promising contribution of an RIS, in the limiting regime
	$M,N \rightarrow \infty$, for enhancing the system performance.
	\vspace*{-0.1cm}
	\subsection{Uplink Ergodic Net Throughput Analysis with a Finite Number of APs and Phase Shifts}
	\vspace*{-0.1cm}
	We now focus our attention on the practical setup in which
	$M$ and $N$ are both finite. By utilizing the user-and-then forget
	channel capacity bounding method \cite{Chien2021TWC}, the uplink ergodic net
	throughput of the user $k$ can be computed in a closed-form
	expression for \eqref{eq:ULRateMRC} as given in Theorem~\ref{theorem:ULMR}. 
	\begin{theorem} \label{theorem:ULMR}
		If the CPU utilizes the MRC method, a lower bound
		closed-form expression for the uplink net throughput of
		the user~$k$ is given as follows
		\begin{equation} \label{eq:ULRateMRC}
			R_{k} = B\nu\left( 1 - \tau_p/\tau_c \right) \log_2 \left( 1 + \mathrm{SINR}_{k} \right), \mbox{[Mbps]},
		\end{equation}
		where $B$ is the system bandwidth measured in MHz and
		$0 \leq \nu \leq 1$ is the portion of each coherence interval that is
		dedicated to the uplink data transmission. The effective uplink
		signal-to-noise-plus-interference ratio (SINR) is
		\begin{equation} \label{eq:ClosedFormSINR}
			\mathrm{SINR}_{k} = \rho \eta_{k} \left( \sum_{m=1}^M \gamma_{mk} \right)^2 \big/ (\mathsf{MI}_{k} + \mathsf{NO}_{k}),
		\end{equation}
		where $\mathsf{MI}_{k}$ is the mutual interference and the noise denoted by $\mathsf{NO}_{k}$ are, respectively, given by
		\begin{align}
			&\mathsf{MI}_{k} = \rho_u \sum_{k'=1}^K \sum_{m=1}^M \eta_{k'} \gamma_{mk} \delta_{mk'} +  p \tau_p \rho_u \sum_{k' \in  \mathcal{P}_k} \sum_{m=1}^M  \eta_{k'} c_{mk}^2   \notag \\
			& \times \mathrm{tr}(\pmb{\Theta}_{mk'}^2) + p \tau_p \rho_u \sum_{k' =1 }^K \sum_{k'' \in \mathcal{P}_k}    \sum_{m=1}^M  \sum_{m'=1}^M\eta_{k'} c_{mk}c_{m'k} \times \notag \\
			& \mathrm{tr}( \pmb{\Theta}_{mk'} \pmb{\Theta}_{m'k''})  + p \tau_p \rho_u \sum_{k' \in \mathcal{P}_k \setminus \{ k\} } \eta_{k'} \left(\sum_{m=1}^M c_{mk} \delta_{mk'} \right)^2,\\
			& \mathsf{NO}_{k} = \sum_{m=1}^M \gamma_{mk},
		\end{align}
		with $\delta_{mk'} =\beta_{mk'} + \mathrm{tr}\big( \pmb{\Phi}^H \mathbf{R}_{m} \pmb{\Phi} \widetilde{\mathbf{R}}_{k'} \big)$,  $c_{mk}$ given in \eqref{eq:cmk}, and $\gamma_{mk}$ given in \eqref{eq:gammamk}.
	\end{theorem}
	\begin{proof}
		The main idea of proof is to average out the randomness by using the use-and-then-forget capacity bounding technique and fundamental properties of Massive MIMO. The detailed proof is available in the journal version \cite{Chien2021TWC}.
	\end{proof}
	By direct inspection of the SINR in \eqref{eq:ClosedFormSINR}, the numerator
	increases with the square of the sum of the variances of the
	channel estimates, $\gamma_{mk},\forall m,$ thanks to the joint coherent transmission.
	On the other hand, the first term in the denominator
	represents the power of the interference. Due to the
	limited and finite number of orthogonal pilot sequences being
	used, it represents the impact
	of pilot contamination. The last term is the additive noise. The
	SINR in \eqref{eq:ClosedFormSINR} is a multivariate function of the matrix of phase
	shifts of the RIS and of the channel statistics, i.e., the channel
	covariance matrices. Compared with conventional Cell-
	Free Massive MIMO systems, the strength of the desired signal
	increases thanks to the assistance of an RIS. However, the
	coherent and non-coherent interference become more severe as
	well, due to the need of estimating both the direct and indirect
	links in the presence of an RIS.
	\begin{figure*}[t]
		\begin{minipage}{0.33\textwidth}
			\centering
			\includegraphics[trim=0.9cm 0cm 1.3cm 0.2cm, clip=true, width=2.4in]{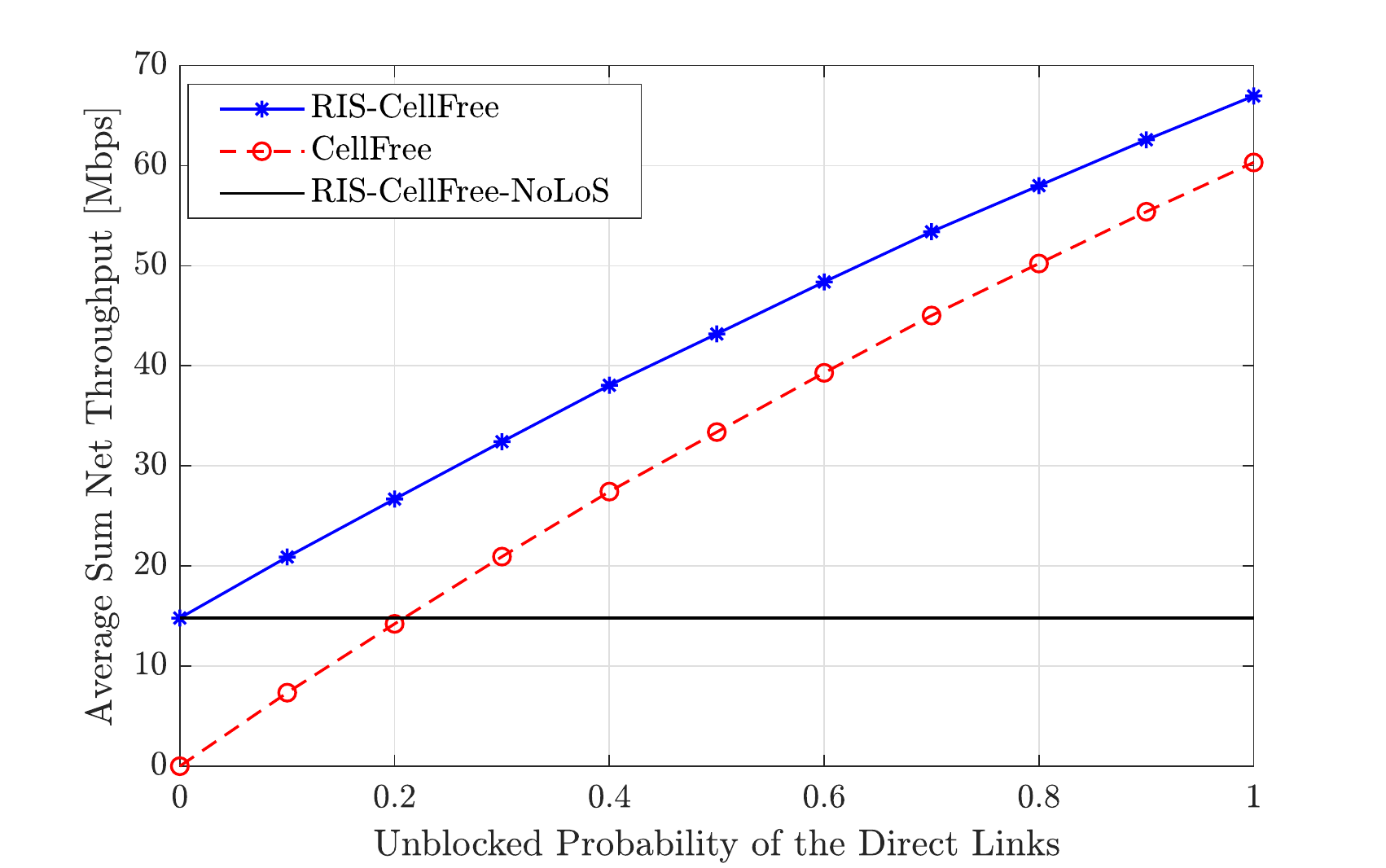} \vspace*{-0.5cm}\\
			(a)
			\vspace*{-0.1cm}
		\end{minipage}
		\begin{minipage}{0.33\textwidth}
			\centering
			\includegraphics[trim=0.9cm 0cm 1.3cm 0.5cm, clip=true, width=2.4in]{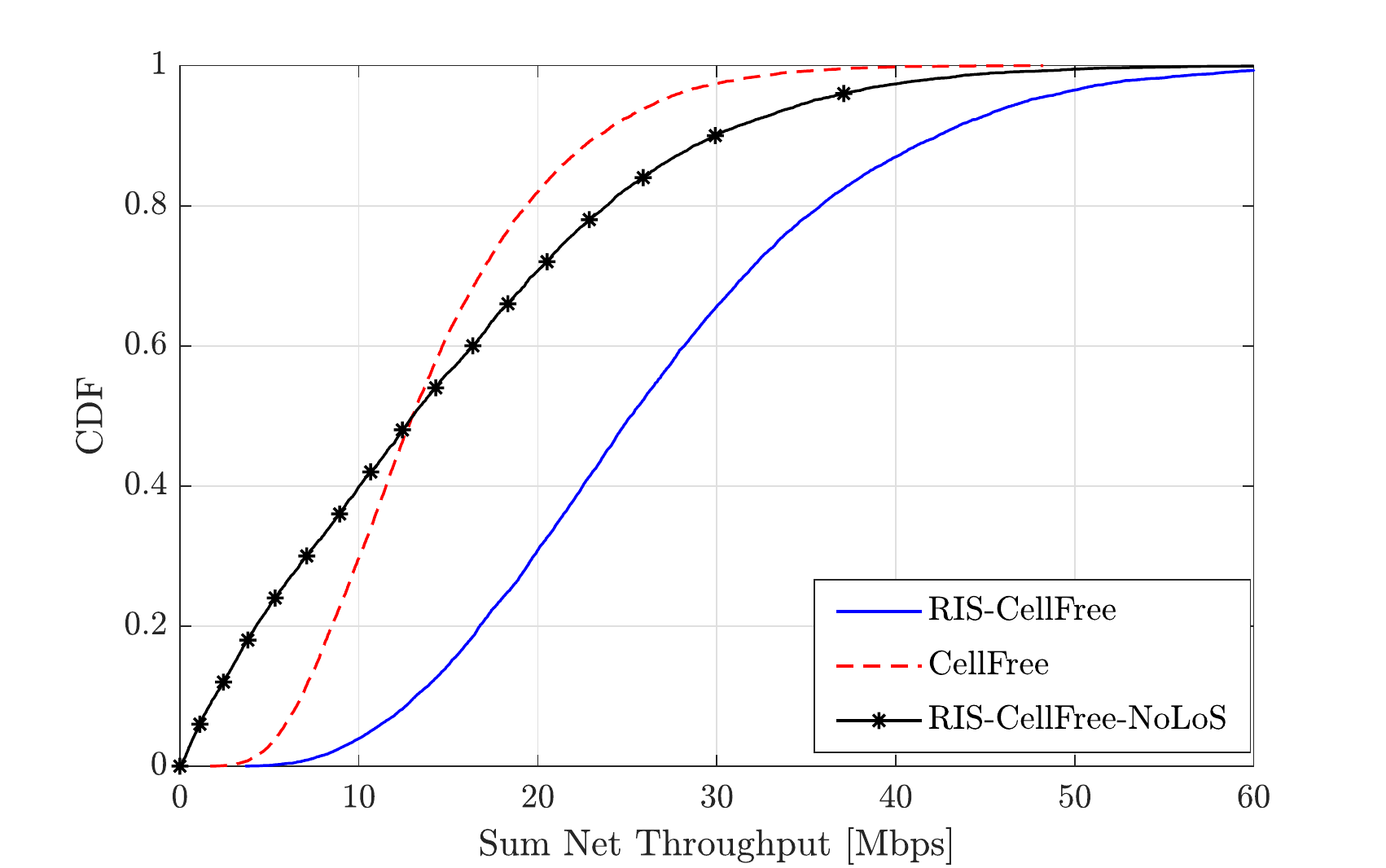} \vspace*{-0.5cm}\\
			(b)
			\vspace*{-0.1cm}
		\end{minipage}
		\begin{minipage}{0.33\textwidth}
			\centering
			\includegraphics[trim=0.9cm 0cm 1.3cm 0.5cm, clip=true, width=2.4in]{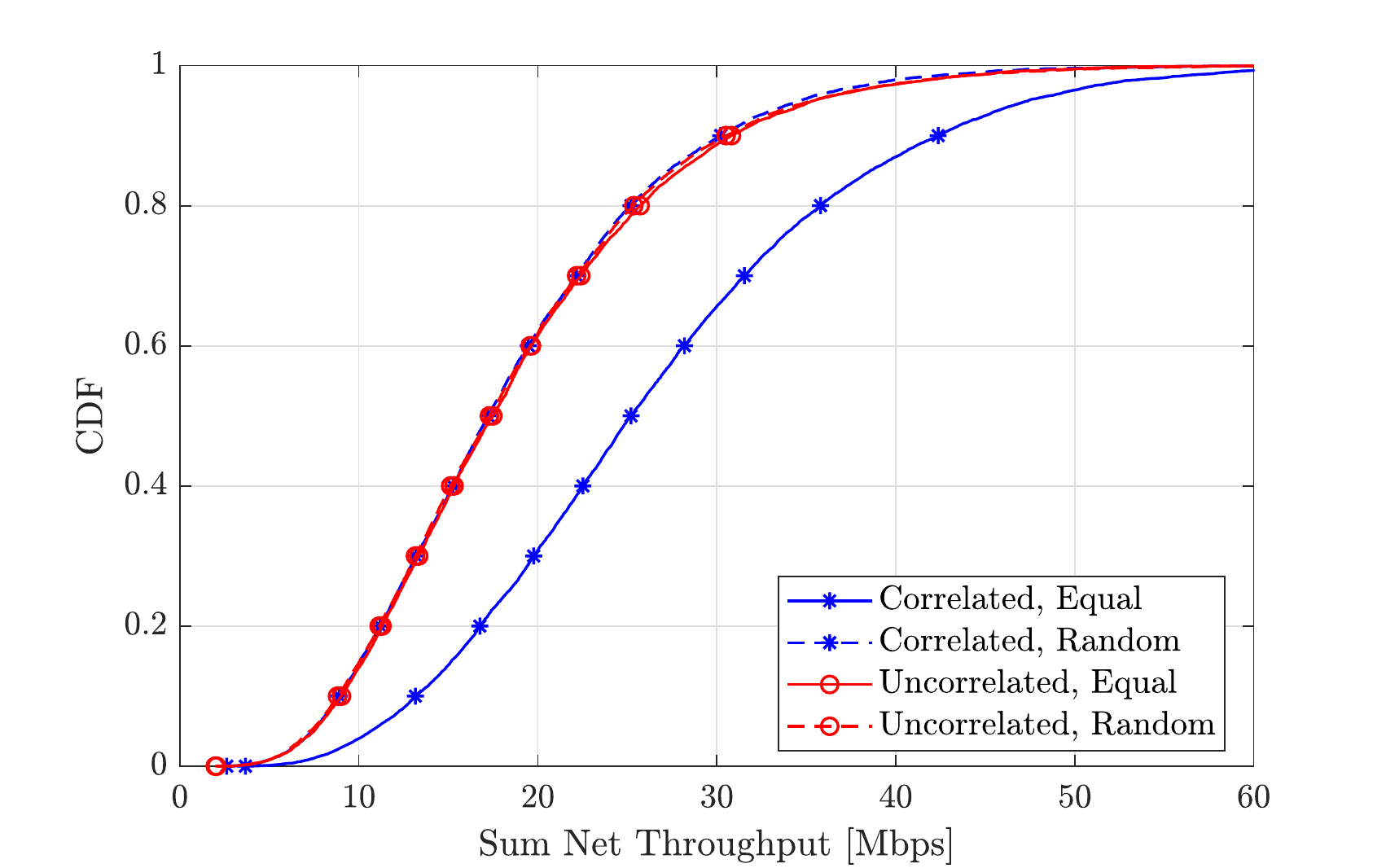} \vspace*{-0.5cm}\\
			(c)
			\vspace*{-0.1cm}
		\end{minipage}
		\caption{The sum net throughput [Mbps]: $(a)$ Average sum net throughput versus the unblocked probability of the direct links; $(b)$ CDF of the sum net
			throughput with the unblocked probability of the direct links $\tilde{p} = 0.2$; $(c)$ Sum net throughput as a function of the phase shift setup (equal or random) and the unblocked probability of the direct links $\tilde{p} = 0.2$.}
		\label{FigPerformance}
		\vspace*{-0.6cm}
	\end{figure*}
	\vspace*{-0.2cm}
	\section{Numerical Results} \label{Sec:NumRes}
	\vspace*{-0.2cm}
	We consider a geographic area of size $1$~km$^2$ that is wrapped
	around at the edges. The locations of $100$ APs
	and $10$ users are given in terms of $(x,y)$ coordinates. To
	simulate a harsh communication environment, the APs are uniformly
	distributed in the sub-region $x,y \in [-0.75, -0.5]$~km,
	while the users are uniformly distributed in the sub-region
	$x,y \in [0.375, 0.75]$~km. An RIS with $N = 900$ is located
	at the origin, i.e., $(x, y) = (0, 0)$. Each coherence interval
	comprises $\tau_c = 200$ symbols and $\tau_p =5$ orthonormal pilot
	sequences. The large-scale fading coefficients $\alpha_m$ and $\tilde{\alpha}_{mk}$ are generated according to the three-slope propagation model in \cite{ngo2017cell}. The large-scale fading coefficient $\beta_{mk}$ is formulated as $\beta_{mk} = \bar{\beta}_{mk}a_{mk},$ where $\bar{\beta}_{mk}$ is generated by the three-slope
	propagation model in \cite{ngo2017cell}. The binary variables $a_{mk}$ accounts for the probability that the direct links are unblocked, and it is
	defined as $a_{mk} =1$ with probability $\tilde{p}$. Otherwise, $a_{mk} =0$
	with probability $1 -\tilde{p}$, where $\tilde{p} \in [0,1]$ is the probability that the direct link is not blocked.The covariance matrices are
	generated according to the spatial correlation model in \eqref{eq:CovarMa} with $d_H = d_V = \lambda/4$. The pilot power is $100$~mW and $\nu = 1$.
	The power control coefficients are $\eta_k = 1, \forall k$. Without loss
	of generality, in particular, the $N$ phase shifts in $\pmb{\Phi}$ are all set
	equal to $\pi/4$, except in Fig.~\ref{FigPerformance}(c) with different phase shifts. Three system configurations are considered for comparison:
	\begin{itemize}
		\item[$i)$] \textit{RIS-Assisted Cell-Free Massive MIMO}: This is the proposed
		system model where the direct links are unblocked
		with probability $\tilde{p}$. It is denoted by ``RIS-CellFree".
		\item[$ii)$] \textit{Conventional Cell-Free Massive MIMO}: This is the
		same as the previous model with the only exception that
		the RIS is not deployed. It is denoted by ``CellFree".
		\item[$iii)$] \textit{Cell-Free Massive MIMO without the direct links}: This is the worst case study in which the direct
		links are blocked with unit probability and the uplink transmission is ensured only through the RIS.
		This setup is denoted by ``RIS-CellFree-NoLOS".
	\end{itemize}
	In Fig.~\ref{FigPerformance}(a), we illustrate the sum net throughput as a function of the probability $\tilde{p}$. In particular, the average sum net throughput is defined as $\sum_{k=1}^K \mathbb{E} \{ R_{k} \}$. Cell-Free Massive MIMO provides the worst performance if the blocking probability is large ($\tilde{p}$ is small). If the direct links are unreliable, as expected, the net throughput offered by Cell-Free Massive MIMO tends to zero if $\tilde{p} \rightarrow 0$. In addition, the proposed RIS-assisted Cell-Free Massive MIMO setup offers the best net throughput, since it can overcome the unreliability of the direct links. An RIS is particularly
	useful if $\tilde{p}$ is small since in this case the direct links are
	not able to support a high throughput. Fig.~\ref{FigPerformance}(b) compares
	the three considered systems in terms of sum net throughput (defined as $\sum_{k=1}^K  R_{k}$)
	when $\tilde{p} = 0.2$. We observe the net advantage of the proposed RIS-assisted Cell-Free Massive MIMO system. The worst-case
	RIS-assisted Cell-Free Massive MIMO system setup (i.e.,
	$\tilde{p} = 0$) outperforms the Cell-Free Massive MIMO setup in the
	absence of an RIS. Fig.~\ref{FigPerformance}(c) focuses  on the RIS-asissted Cell-Free Massive MIMO setup, since it provides
	the best performance. We compare the sum net throughput as
	a function of the phase shifts of the RISs (random and uniform
	phase shifts according to Corollary \ref{corollary:EqualPhase} in the presence of
	spatially-correlated and spatially-independent fading channels
	according to \eqref{eq:CovarMa}. In the presence of spatial correlation, we consider
	$\mathbf{R}_m = \alpha_m d_H d_V \mathbf{I}_N$ and $\widetilde{\mathbf{R}}_k = \tilde{\alpha}_{k} d_H d_V \mathbf{I}_N, \forall m,k$. If
	the spatial correlation is not considered, there is no significant
	difference between the random and uniform phase shifts setup.
	In the presence of spatial correlation, on the other hand, the
	proposed uniform phase shift design, which is obtained from
	Corollary~\ref{corollary:EqualPhase}, provides a much better throughput. This result
	highlights the relevance of using even simple optimization
	designs for RIS-assisted communications in the presence of
	spatial correlation.
	\vspace*{-0.2cm}
	\section{Conclusion}\label{Sec:Conclusion}
	\vspace*{-0.2cm}
	We have considered an RIS-assisted Cell-Free Massive
	MIMO system and have introduced an efficient channel estimation scheme to overcome the high channel estimation
	overhead. An optimal design for the phase shifts of
	the RIS that minimizes the channel estimation error has been
	introduced and has been used for system analysis. Also, a closed-form
	expression of the ergodic net throughput for the uplink data transmission phase has been proposed. Based on them, the performance of RIS-assisted Cell-Free Massive MIMO has
	been analyzed as a function of the fading spatial correlation
	and the blocking probability of the direct AP-user links. The
	numerical results have shown that the presence of an RIS is
	very useful if the AP-user links are mostly unreliable with high probability.
	\vspace*{-0.2cm}
	\bibliographystyle{IEEEtran}
	\bibliography{IEEEabrv,refs}
	\vspace*{-0.4cm}
\end{document}